\documentclass[11pt]{article}

\usepackage[margin=1in]{geometry}
\usepackage[utf8]{inputenc}
\usepackage{amsmath,amssymb,amsthm,amsfonts,setspace,mathtools}
\usepackage[shortcuts]{extdash}
\usepackage{graphicx}
\usepackage{chngcntr}
\usepackage{apptools}
\usepackage{nicefrac}

\let\oldFootnote\footnote
\newcommand{\nextToken}{\relax}
\renewcommand{\footnote}[1]{\oldFootnote{#1}\futurelet\nextToken\isFootnote}
\newcommand{\isFootnote}{\ifx\footnote\nextToken\textsuperscript{,}\fi}

\usepackage{color}
\usepackage{hyperref}
\usepackage{xcolor}
\definecolor{dark-red}{rgb}{0.4,0.15,0.15}
\definecolor{dark-blue}{rgb}{0.15,0.15,0.4}
\definecolor{medium-blue}{rgb}{0,0,0.5}
\hypersetup{
    colorlinks, linkcolor={dark-red},
    citecolor={dark-blue}, urlcolor={medium-blue}
}
\usepackage{tikz-cd}
\usepackage[all]{xy}

\newtheorem{theorem}{Theorem}

\newtheorem{corollary}{Corollary}
\newtheorem{lemma}{Lemma}
\newtheorem{proposition}{Proposition}

\AtAppendix{\counterwithin{lemma}{section}}
\AtAppendix{\counterwithin{proposition}{section}}
\AtAppendix{\counterwithin{assumption}{section}}
\AtAppendix{\counterwithin{corollary}{section}}
\AtAppendix{\counterwithin{definition}{section}}
\AtAppendix{\counterwithin{theorem}{section}}
\AtAppendix{\counterwithin{claim}{section}}
\AtAppendix{\counterwithin{example}{section}}
\AtAppendix{\counterwithin{equation}{section}}
\AtAppendix{\counterwithin{remark}{section}}

\usepackage{caption}
\usepackage[labelformat=simple]{subcaption}

\theoremstyle{definition}
\newtheorem{definition}{Definition}

\theoremstyle{remark}

\newcommand{\zerodel}{.\kern-\nulldelimiterspace}

\renewcommand\emptyset\varnothing

\newcommand\ssm\smallsetminus

\newcommand\subtype{matched type}

\newcommand\multisetDef{X^{(\leq N)}}

\newcommand{\msubseteq}{\sqsubseteq}
\newcommand{\msubset}{\sqsubset}

\newcommand{\mcup}{\sqcup}
\newcommand{\mbigcup}{\bigsqcup}

\newtheorem{innercustomgeneric}{\customgenericname}
\providecommand{\customgenericname}{}
\newcommand{\newcustomtheorem}[2]{%
  \newenvironment{#1}[1]
  {\renewcommand\customgenericname{#2}%
   \renewcommand\theinnercustomgeneric{\ref*{##1}$'$}%
   \innercustomgeneric
  }
  {\endinnercustomgeneric}
}
\newcustomtheorem{primetheorem}{Theorem}

\theoremstyle{definition}

\newtheorem{innercustomgenericdef}{\customgenericname}
\providecommand{\customgenericname}{}
\newcommand{\newcustomdef}[2]{%
  \newenvironment{#1}[1]
  {\renewcommand\customgenericname{#2}%
   \renewcommand\theinnercustomgenericdef{\ref*{##1}$'$}%
   \innercustomgenericdef
  }
  {\endinnercustomgenericdef}
}
\newcustomdef{primedefinition}{Definition}

\usepackage{setspace}
\usepackage[margin=1in]{geometry}
\usepackage[margin=1in]{geometry}
\usepackage[round]{natbib}

\title{Many-to-many stable matching in large economies}
\author{
Michael Greinecker\thanks{CEPS, ENS Paris-Saclay.  Email: {\tt michael.greinecker@ens-paris-saclay.fr}. } \and
Karolina Vocke\thanks{Department of Economics, University of Innsbruck.  Email: {\tt karolina.vocke@gmx.de}. }}
\begin{document}
\maketitle
\abstract{We study stability notions for networked many-to-many matching markets with individually insignificant agents in distributional form. Outcomes are formulated as joint distributions over characteristics of agents and contract choices. Characteristics can lie in an arbitrary Polish space. We provide a mechanical method for transferring existence results for finite matching models to large matching models for many stability notions. In particular, we show that tree-stable and pairwise-stable outcomes exist.}

\section{Introduction}

This paper provides a toolkit for transferring existence results for stability concepts from finite matching markets or matching markets with finitely many types to markets with large type spaces and (implicitly) a continuum of agents. We consider many-to-many matching markets modeled via the joint distribution of the characteristics of the agents involved and show that existence results that hold for finite type spaces also hold in Polish (that is, separable and completely metrizable) spaces, including compact metrizable or Euclidean type spaces.\\

The theory of stable matching markets examines allocation problems where the identity of a partner matters, and markets operate without frictions. \citet{gale1962college} introduced the ``marriage market'' as the simplest model of a two-sided matching market. This model, and many-to-one generalizations, have been applied to solve allocation problems such as the allocation of students to colleges, assistant physicians to hospitals, and workers to jobs.  
In more complex markets, when multiple contracts can be signed on both sides or agents interact along complex networks, one faces the more complex structure of networked many-to-many markets.  In one-to-one matching markets, stability requires
the absence of a mutually desirable unsigned ``blocking'' contract between any pair of agents.
However, there is no single, canonical extension of this idea to many-to-many markets. Different notions of stability have been introduced that can differ in the shape of blocks or what it means to be a block; see \citet{klaus2009stable}. For most of the existing stability concepts, including pairwise stability, existence can only be guaranteed under strong restrictions on the network topology or on preferences. The roommates problem shows that not even pairwise stable outcomes are guaranteed to exist in finite one-to-one matching markets. Given this variety of stability notions, there is also a variety of existence results in finite markets and, more recently, in continuum models with finitely many types. \\ 

A recent literature on large matching markets has shown that stability concepts are better behaved in markets with a continuum of agents. Strong assumptions, such as substitutability assumption, familiar from models with finitely many agents can be substantially weakened or dropped in large markets to guarantee the existence of specific stability concepts; see \citet{azevedo2015existence}, \citet*{kojima2013matching}, and \citet{jagadeesan2021stability}. Apart from technical considerations,
\citet{vocke2020setwise} argues that most stability notions for many-to-many matching markets apply most naturally to markets that have many agents to begin with. In this paper, we are agnostic about the appropriate notion of stability and develop a tool that allows one to transfer any existence result for any stability notion from models with finitely many types or agents to markets with a continuum of agents and arbitrary types. \\

All the papers mentioned so far assume that there are only finitely many types or contracts. However, for many applications, this assumption is too restrictive. 
Models with transferable utility, for example, cannot be described with finite contract spaces. Typical distributional assumptions on, say, the income distribution (such as a Pareto distribution) require continuous distributions and \emph{a fortiori} a continuum of characteristics. Allowing for large contract and type spaces, however, brings new issues in defining appropriately adapted stability notions and proving existence.\\ 

In our model, there are Polish type and contract spaces $T$ and $Y$ with a population distribution on $T$. Outcomes are defined to be joint distributions over types and contract choices. Canonical examples of Polish spaces are compact metric spaces and the Euclidean space $\mathbb{R}^n$. Agents in our model are allowed to sign a finite number of contracts that are feasible for them under a continuous contract correspondence. To be able to use methods from topological measure theory, the space of contract choices has to be suitably topologized. Suppose that  an agent can choose two contracts from the unit interval. A sequence of such choices containing the contracts $0$ and $1/n$ should converge in the limit to a choice that still involves two contracts. Intuitively, the contract $0$ is chosen twice in the limit. This requires working with multisets of contracts, and we develop a novel topological theory of Polish multispaces (in Appendix \ref{app:multi}) to deal with the issue. \\

To reformulate existing stability notions to fit our model, we adapt an approach that has originally been introduced by \citet{greinecker2018pairwise} in the context of one-to-one matching models.
An outcome is stable under any stability notion that works with 
finitely many types, if the probability of finding a block by sampling independently from an outcome is zero. Our central result is that for these transferred stability notions, any existence result transfers as well. This crucially relies on the continuity of the utility function defined on multisets of contracts and the continuity of the contract correspondence. Combining our transfer result with the existence result of \citet{jagadeesan2021stability}, we can show that, in particular, pairwise and tree-stable outcomes exist in large markets with general type and contract spaces.\\

The remainder of this paper is organized as follows. In Section 2, we introduce our model and discuss different stability concepts. In Section 3, we state our main existence result and some auxiliary lemmas together with proof sketches. In Section 4, we give some illustrative examples to show what matching markets in large type spaces can look like. Section 5 concludes.

\paragraph{Related Literature}



Our work relies heavily on the work of \citet{greinecker2018pairwise}. 
They introduce a distributional matching model 
and consider a two-sided one-to-one continuum market with general type spaces. 
  \citet{greinecker2018pairwise} show that stable outcomes exist as a consequence of the existence of pairwise stable outcomes in finite two-sided one-to-one markets. They also allow for externalities and show that the continuum model can be approximated by finite markets. \citet{carmona2024stable} extend the model to allow for one-to-many matching and the option for agents to choose their occupation.
  The specific definition of a networked many-to-many market and the definition of stability concepts in our model rely on \citet{jagadeesan2021stability}, which generalizes \citet{azevedo2015existence}
  large-market model to networks. \citet{ostrovsky2008stability} models matching in networks, but required contracts to have natural ``buyer'' and ``seller'' counterparties. He introduces the concept of tree-stability in finite markets, which is adapted to a continuum of agents in \citet{jagadeesan2021stability}.\\ 

  There is a huge literature delivering existence results for matching markets under different restrictions on the market structure or on preferences. In finite markets, \citet{gale1962college} show that stable outcomes in two-sided one-to-one markets always exist. They also show that, if agents have substitutable preferences, stability and pairwise stability coincide, and stable outcomes exist in two-sided many-to-many markets.
  \citet*{che2019stable} show the existence of stable outcomes in large-market many-to-one settings with complementarities. Their model is extended to allow for externalities by \citet{CarmonaLaohakunakorn2023}. \citet{azevedo2015existence} show that under a one-sided substitutability condition, stable outcomes exist in two-sided many-to-many markets with a continuum of agents and finitely many types. \citet{jagadeesan2021stability} show that in networked many-to-many markets with a continuum of agents and finitely many types, tree-stable, hence, pairwise stable outcomes always exist for general preferences.\\
  
  In continuum one-to-one matching markets, \citet*{goz92} and \citet*{CMNTransport} already proved the existence of stable matchings with transferable utility using the duality theory of optimal transport. The arguments are specific to the transferable utility setting. \citet{noldeke2015implementation} prove existence of stable matchings with imperfectly transferable utility. They, too, approximate their continuum model with a finite model and use a compactness argument to obtain a stable matching for the limit economy. However, their compactness argument operates in the joint space of measures and utility functions.
  
  





\section{Model}

Let $T$, $Y$ be Polish spaces.
We interpret $T$ as the space of types an agent can have, with population distribution $\nu\in\Delta(T)$,\footnote{For $X$ a Polish space, $\Delta(X)$ represents the space of Borel probability measures on $X$, topologized with the topology of weak convergence. Products of Polish spaces are endowed with their product topologies, which are, again, Polish.} and $Y$ as a space of contracts. For each contract, there is an additional specification of two sides, given by $X=Y\times\{1,2\}$. Contracts may not be symmetric because they might require different roles from the contract partners, such as being a buyer or seller. However, the same agent could take different sides for different contracts. We will define a contract correspondence by the restriction to signing at most $N$ contracts and a restriction of which contracts are feasible between two agents of certain types.  Not all elements in $X$ are feasible for every pair of agents. The set $\bar{\chi}(t,t')\subseteq Y$ describes the possible contracts a pair of agents of types $(t,t')$ can sign. Formally, there is a continuous compact-valued correspondence $\bar{\chi}:T\times T\rightrightarrows Y$, such that a contract $y$ lies in $\bar{\chi}(t,t')$ if agents of type $t$ can sign the contract $(y,1)$ and agents of type $t'$ can sign the contract $(y,2)$, respectively. Thus, for an agent of type $t$, the set of feasible contracts in $X$ is given by \[X_t=\{(y,1)\mid y\in \bar{\chi}(t,t')\text{ for some }t'\in T\}\cup \{(y,2)\mid y\in \bar{\chi}(t',t)\text{ for some }t'\in T\}.\]
We assume that individual agents can sign a maximum of $N$ contracts. The option of agents signing some of the contracts multiple times is captured in our model by considering \emph{multisets} instead of sets. For $n\in \mathbb{N}$ the set $X^{(n)}$ is the set of possible multisets with $n$ elements in $X$. Each such multiset is described by a function $m:X\to \mathbb{N}$ that has value $0$ at all but finitely many points and such that $\sum_{(x,l)\in m}l=n$.  
The set $\multisetDef=\bigcup_{n\leq N}X^{(n)}$ is the set of possible multisets of contracts with at most $N$ elements. The set $X^{(\leq N)}$ is equipped with a suitable Polish topology inherited from the topology on $X$; see Appendix \ref{app:multi}. A multiset $m$ is a \emph{multisubset} of a multiset $m'$, written $m\msubseteq m'$, if $m(x)\leq m'(x)$ for all $x\in X$.\\

Signing no contract is always feasible.
We use $\bar{\chi}$ to derive the \emph{contract correspondence} $\chi:T\rightrightarrows\multisetDef$, which specifies the set of possible multisets of contracts an agent $t$ can sign and is given by $\chi(t)=X_t^{(\leq N)}$ with $\emptyset=X^{(0)}\in\chi(t)$ for all $t\in T$. \\

\emph{Preferences} over multisets of contracts are
represented by a continuous utility function $u:T\times \multisetDef\to\mathbb{R}$. We write $u_t:\multisetDef \to\mathbb{R}$ for the t-section.\footnote{Implicitly, a type has preferences over multisets of contracts they could not participate in. However, we really need $u$ only defined on the closed set of such feasible type-contract-combinations. The function can then be extended by the Tietze extension theorem to the ambient space. The specific extension will be irrelevant for the arguments we make.}\\

We call the quintuple $(T,Y,\nu,\chi,u)$ a \emph{matching problem}. The matching problem is \emph{finite} if $T$ and $Y$ are finite and $\nu$ has only rational values. A finite matching problem is one that could be realized as the population distribution of a finite number of agents. 


For $\tau\in\Delta(V\times W)$, we write $\tau_V$ and $\tau_W$ for the respective marginals. If $\tau$ is a measure on a multispace, we write $\tau^*$ for the measure on the underlying space that counts sets with the corresponding multiplicity. We use it to go from a measure specifying how often certain  contract choices are made, a measure over multisets, to the measure specifying how often each contract is signed. The formal definition of this $*$--mapping is somewhat involved and given in Appendix \ref{app:multi}.

 
The following definition adapts the usual definition of outcomes from finite markets to a distributional setting.

 \begin{definition}
 \label{def:outcome}
  An \emph{outcome} for $\nu$ is a probability measure $\mu\in \Delta\big(T\times \multisetDef\big)$ such that
  \begin{enumerate}
  \item The $T$-marginal of $\mu$ is $\nu$.
   \item For $\mu$-almost all $(t,\mathcal{Y})$, we have $\mathcal{Y}\in\chi(t)$.

  \item 
  
  For every Borel set $E\subseteq Y$, we have
 \[\mu_{\multisetDef}^*(E\times\{1\} )=\mu_{\multisetDef}^*(E\times\{2\}).\]
  \end{enumerate} 
\end{definition}
The first condition guarantees that the implicit type distribution in an outcome agrees with the given population distribution, the second condition guarantees that nobody signs unavailable contracts, and the third condition guarantees that every contract is signed by the same number of counter-parties on both sides.

\subsection{Stability}
\paragraph{Blocking structures} Intuitively, a block consists of a group of finitely many agents that can improve over a given outcome by getting together and signing new contracts, while possibly dropping existing contracts. Stability is defined by the absence of any blocks. However, in a large market it is not entirely clear what it means that a block arises. Formally, agents are not data of the model; only the joint distribution of types and contract choices is specified in an outcome. However, we can sample agents by repeated draws from an outcome. We consider the outcome stable if the probability of such a draw giving rise to a block is zero. Since only the type and the set of existing contracts matters for the incentive to deviate, we define blocks on the level of matched types. Formally, a \emph{\subtype} $(t,\mathcal{Y})$ consists of a type $t \in T$ and a multiset $\mathcal{Y} \in \multisetDef$.\\

If single agents have an incentive to unilaterally drop existing contracts, they can form an individual block all by themselves.
\begin{definition}
   An individual block consists of a \subtype\  $(t,\mathcal{Y})$ such that there exists $\mathcal{W}\msubset \mathcal{Y}$ with $u_t(\mathcal{W})>u_t(\mathcal{Y})$.
\end{definition}

To define stability for groups, we now define more complex blocks with the shapes of arbitrary graphs.  We then restrict attention to specific shapes of graphs to define tree-stability and pairwise stability. We consider graphs $G$ whose set of vertices is $\{1,2,\ldots,n\}$ and specify directed graphs by their sets of (directed) edges.
Formally, a \emph{graph with $n$ vertices} is a family $G$ of pairs from $\{1,2,\ldots,n\}$; the members of $G$ are called \emph{edges}.
A graph is a \emph{tree} if there exists a unique path between each pair of vertices.\footnote{
A \emph{path} between the vertices $j$ and $k$ in  $G$ is a finite sequence $(j,n_1), (n_1,n_2),\ldots,(n_m,k)$ in $G$.
}

Instead of identifying the participating agents directly, we specify their \subtype{s} $(i_1,\mathcal{Y}^1),\ldots,(i_n,\mathcal{Y}^n)$---which are sufficient to determine whether the agents want to deviate.\footnote{Note that the \subtype{s} $(i_1,\mathcal{Y}^1),\ldots,(i_n,\mathcal{Y}^n)$ involved in a block need not be distinct.}

\begin{definition}
\label{def:block}
  Let $G$ be a simple directed graph with $n$ vertices.
A block of shape $G$ consists of
\begin{itemize}
    \item[-] a \subtype\ $(t_j,\mathcal{Y}_j)$ for each vertex $1 \le j \le n$, and
    \item[-] a contract $x_{(k,\ell)} \in Y$ for each edge $(k,\ell) \in G$
\end{itemize}
for which
\begin{itemize}
\item[-] writing

\[\mathcal{Z}^j = \mbigcup_{k\mid (j,k)\in G}(x_{(j,k)},1) \mcup \mbigcup_{k\mid (k,j)\in G}(x_{(k,j)},2),\]

2n
we have that there exists a $\mathcal{W}^j\msubseteq \mathcal{Z}^j\mcup \mathcal{Y}^j$ such that $\mathcal{W}^j\in \chi(t_j)$, $\mathcal{Z}^j\msubseteq \mathcal{W}^j$, and $u_{t_j}(\mathcal{W}^j)>u_{t_j}(\mathcal{Y}^j)$ for each $1\leq j \leq n$. 
\end{itemize}




\end{definition}

The last condition says that for each \subtype\  $(t_j,\mathcal{Y}_j)$, there exists a feasible multiset of contracts $\mathcal{W}^j$
that contains all newly proposed contracts $\mathcal{Z}^j$ and possibly some of the \subtype’s existing contracts $\mathcal{Y}^j$, and that yields a strict utility improvement over the current
contract multiset $\mathcal{Y}^j$.

\paragraph{Stability in finite type and contract spaces}
In this section, we introduce stability concepts for a finite type space $T_F$ and a finite contract space $Y_F$.\footnote{Still, we do not restrict $\nu$ to have rational values; the population distributions need not be the type distribution of a finite number of agents.} 
In a large market of this kind, an outcome is stable if no positive measure of agents has an incentive to \emph{block} the outcome. Intuitively, a finite number of agents cannot block an outcome, so we require that there are positive masses of agents of each of the participating \subtype{s} for a block to arise at $\mu$. \\

An individual block for a \subtype\ $(t,\mathcal{Y})$ \emph{arises} at an outcome $\mu$ if the \subtype\ has positive measure, i.e. $\mu(t,\mathcal{Y})>0$.

\begin{definition}
\label{def:indrat}
   An outcome is \emph{individually rational} if no individual block arises.
\end{definition}

Let $n>1$ and $G$ be a block with $n$ vertices. A $G$-block \emph{arises} at an outcome $\mu$ if $\mu(t_j,\mathcal{Y}^j) > 0$ for all $1\leq j\leq n$, with $(t_j,\mathcal{Y}^j)$ the matched type assigned to the vertex $j$.\\

By varying the possible shapes of the underlying graphs of blocks, one can specify various known stability notions.
\begin{definition}
\label{def:stab}
Let $\mathcal{G}$ be a family of graphs. An outcome $\mu$ is \emph{$\mathcal{G}$-stable} if $\mu$ is individually rational and no $G$-block arises for any 
$G\in\mathcal{G}$.

\begin{itemize}
\item[-] If $\mathcal{G}$ is the family of graphs with two vertices and one edge, then we call a $\mathcal{G}$-stable outcome \emph{pairwise stable}. 
    \item[-] If $\mathcal{G}$ is the family of all trees, then we call a $\mathcal{G}$-stable outcome \emph{tree-stable}.
    \item[-] If $\mathcal{G}$ is the family of all graphs, then we call a  $\mathcal{G}$-stable outcome \emph{stable}. 
\end{itemize}%

\end{definition}
While we do not formally allow for multiple contracts between pairs of agents in a block, this restriction does not affect the definition of stability in our large-market model---as shown by \citet{jagadeesan2021stability} in their Appendix B1.\\

We are going to leverage the main existence result of \citet{jagadeesan2021stability}. \citet{jagadeesan2021stability} show that tree-stable outcomes exist whenever all type and contract spaces are finite; see Lemma \ref{lem:TYfinite} for the statement of their theorem in the language of our model.
Note that if the underlying network of possible contracts is acyclic, every tree-stable outcome is trivially also stable. Less trivially, it is shown by  \citet{jagadeesan2021stability} that if preferences are substitutable, then every tree-stable outcome is stable.

\paragraph{Stability with large type spaces}

In this section, we adapt the definition of stability to potentially large type and contract spaces. Let $T$ and $Y$ be arbitrary Polish spaces. Intuitively, we test stability via random samples from the outcome.\\

\noindent An \emph{$n$-sample} is a finite sequence of \subtype s $t^n=(t_1,\mathcal{Y}_1),\ldots, (t_n,\mathcal{Y}_n)$ of length $n$.
By abuse of notation, we call an $n$-sample a $G$-block if there is a block of shape $G$ with $n$ vertices with the assignment of the \subtype s to vertices given by $t^n$.\\

\noindent An individual block arises at an outcome if 
the set of individual blocks has positive $\mu$-measure. 
A block of shape $G$ with $n$ vertices \emph{arises} at the outcome $\mu$, if the set of $G$-blocks has positive $\otimes^n \mu$-measure, where $\otimes^n \mu$ is the $n$-fold product measure derived from $\mu$. \\

\noindent An outcome $\mu$ is $\mathcal{G}$-stable if no block of shape $G$ arises for any $G\in\mathcal{G}$. If type and contract spaces are finite, this is equivalent to the previous stability definition.


\section{Existence in large type spaces}

This section contains our central existence result, which shows that $\mathcal{G}$-stable outcomes exist in the general model with Polish type and contract spaces, whenever $\mathcal{G}$-stable outcomes exist for finite type and contract spaces. In particular, tree-stable outcomes exist in full generality.\\

\noindent Our existence argument adapts the one of \citet{greinecker2018pairwise}. We first go from finite type and contract-spaces to finite type spaces but general contract spaces. We approximate the infinite contract spaces by smaller finite subspaces for which stable outcomes exist by assumption. A compactness argument allows us to take a convergent subsequence that converges to an outcome for the general space of contracts, and that outcome is, by a continuity argument, stable. We then approximate the general type distributions with distributions of types that have finite support. We choose stable matchings for those and, similarly, employ again a compactness argument to pass to an outcome for the limit type distribution. Again, the resulting matching will be stable. The central lemmas guarantee the needed compactness in Lemma \ref{Lemmacompact} and Lemma \ref{tight}, and the openness of the space of $G$-blocked samples in Lemma \ref{open}. The latter is what is needed for our continuity arguments.\\

\noindent For our compactness argument, we require an additional assumption. In general, the correspondence $\chi$ need not have a closed graph, which we will require for our existence proof. It is not hard to show that the correspondence $\chi$ has a closed graph if and only if the correspondence $t\mapsto X_t$ has a closed graph. To guarantee that the correspondence $t\mapsto X_t$ has a closed graph and to simplify a further compactness argument, we assume there is an upper hemicontinuous and compact-valued correspondence $P:T\rightrightarrows T$ such that $\bar{\chi}(t,t')=\emptyset$ for $t'\notin P(t)$. Let $P_1$ be the correspondence given by $t\to \{t\}\times P(t)$ and $P_2$ be the correspondence given by $t\to P(t)\times\{t\}$. Then, 
\[X_t=(\bar{\chi}\circ P_1)(t)\times\{1\}\cup(\bar{\chi}\circ P_2)(t)\times\{2\},\]
which is a compact-valued upper-hemicontinuous correspondence by \citet[Theorem 17.23, Theorem 17.27, Theorem 17.28]{AliprantisBorder2003} and has, therefore, a closed graph by \citet[Theorem 17.10]{AliprantisBorder2003}. Under this assumption, each type can only sign contracts with a compact set of types. However, this assumption holds without loss of generality if contracts with types outside $P(t)$ could not be individually rational for type $t$. While this assumption is automatically satisfied if $T$ itself is compact, this model can also be applied to more general situations without compact type spaces.

The following three lemmas take care of some technical details in the main existence proof.
 \begin{lemma} 
 \label{Lemmacompact}The set
\[\big\{(\nu,\mu)\in\Delta(T)\times\Delta(T\times \multisetDef)\mid \mu\textnormal{ is an outcome for }\nu\big\}\]
is closed. 
\end{lemma}

\begin{lemma} \label{tight}
 Let $\langle \nu_n\rangle$ be a sequence in $\Delta(T)$ converging to $\nu$ and let $\langle \mu_n\rangle$ be a sequence of corresponding outcomes. Then a subsequence converges to an outcome $\mu$ for $\nu$.
\end{lemma}


\begin{lemma}\label{open}
For each shape $G$ with $n$ vertices the set of $G$-blocked $n$-samples is open.
\end{lemma}

The following is our main result. The theorem says that existence results for various stability notions can be transferred to large type and contract spaces. 

\begin{theorem}
\label{thm:exist}If for all finite matching problems with type and contract spaces $T_F\subseteq T,Y_F\subseteq Y$ a $\mathcal{G}$-stable outcome exists, then a $\mathcal{G}$-stable outcome exists.
\end{theorem}

It follows from the proof below that the argument also works for finite type spaces. This can, for example, be useful for extending transferable utility models with finitely many types to large type spaces. Also, one can allow for any set of population distributions that is dense in the weak topology.


\begin{proof}
Let $G\in \mathcal{G}$ be a graph with $n$ vertices. We first show that if a $G$-stable outcome exists for the finite matching problem $(T_F,Y_F,\nu_F,\chi_F,u_F)$ then a $G$-stable outcome exists. The theorem then follows from $\mathcal{G}$ being countable and the fact that the countable union of measure zero sets has measure zero.\\
We first assume $\nu$ to have finite support and only rational values. For a sequence of contracts $y_1,y_2\ldots,$ that is dense in $Y$ with $y_1=\o$ we define $Y_m=\{y_1,\ldots,y_m\}$ and let $\chi_m(t):=\chi(t)\cap Y_m^{(\leq N)}$. Hence, the matching problems $(T,\nu,Y,\chi_m,u)$ are finite and, by assumption, there exists a $G$-stable outcome $\mu_m$ for each $m$. We can assume (by possibly passing to a subsequence) that the sequence $\langle\mu_m\rangle$ converges to some outcome $\mu$ by Lemma \ref{Lemmacompact} and Lemma \ref{tight}. We show that $\mu$ is $G$-stable for the unrestricted correspondence $\chi$. 
The set of $n$-samples that are $G$-blocked by blocks using only contracts in $Y_m$ is open, and the union over all $m$ gives the set of all samples that are $G$-blocked, which is open. Each of these countably many sets must have eventually $\otimes^n \mu_m$-measure zero, and hence, by the Portmanteau theorem, also $\otimes^n \mu$ measure zero. Since the countable union of measure zero sets has measure zero, $\mu$ is $G$-stable. \\ 

We now allow for $\nu$ to be general. The set of probability measures with finite support and rational values is weakly dense in $\Delta(T)$, so there exists a sequence $\langle \nu_m\rangle$ of such measures that converges weakly to $\nu$. By what we have shown above, for each such $\nu_m$, there exists a $G$-stable outcome $\mu_m$. Using again Lemma \ref{Lemmacompact} and Lemma \ref{tight}, we can assume, by possibly passing to a subsequence, that $\langle \mu_m\rangle$ converges to some outcome $\mu$. We show that $\mu$ is $\mathcal{G}$-stable. 
The set of $G$-blocked $n$-samples $O$ is open by Lemma \ref{open}. The sequence $\langle \otimes^n \mu_m\rangle$ converges weakly to $\otimes^n \mu$ by \citet[Theorem 2.8]{billing}. Therefore, by the Portmanteau theorem, 
\[0=\liminf_m \otimes^n \mu_m(O)\geq\otimes^n \mu(O)=0.\]
No $G$-block arises at $\mu$.
\end{proof}

\noindent The following lemma is a consequence of the existence result in \citet{jagadeesan2021stability}; see Appendix \ref{sec:proofs} for how to translate their theorem to our model. 

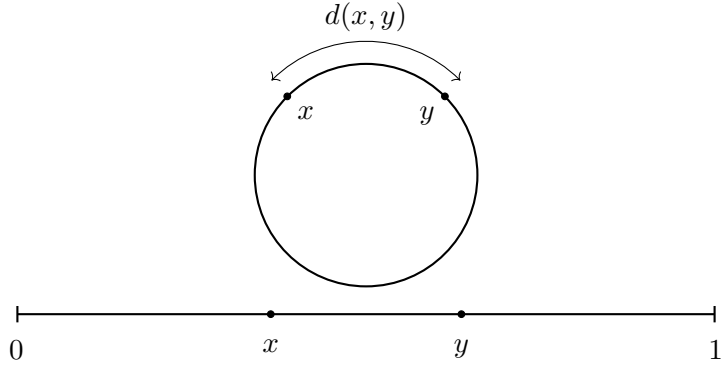
\begin{lemma}[Jagadeesan and Vocke]
\label{lem:TYfinite}
If $T$, $Y$ are finite sets, then a tree-stable outcome exists.
\end{lemma}
Combining Lemma \ref{lem:TYfinite} with Theorem \ref{thm:exist} gives us the following corollary.

\begin{corollary}\label{treeexist}
Tree-stable and, in particular, pairwise stable outcomes exist for general type and contract spaces.
\end{corollary}



\section{Illustrative Examples}

In the following section, we give two examples to illustrate how our framework and stability notion can be applied. The examples are chosen to discuss type spaces that can only be modeled within our framework. The type spaces in the examples are the interval and the circle, two of the easiest continuum type spaces one could think of. 


\paragraph{A roommate problem with assortative preferences:} 
Let types be given by $T=[0,1]$ , and contracts be uniquely defined by the two types involved, hence $Y=T$, and $\chi(t)=X^{(\leq 1)}$. All agents have the same preferences and prefer agents of higher types while being indifferent about the role within the contract $(t, 1)$ or $(t, 2)$. Thus, the utility for type $t$ to be matched with type $t'$ is given by $u_t((t',1))=u_t((t',2))=t'$ and $u_t(\emptyset)=-1$.\\

\begin{figure}
\centering
\vspace{2cm}
\begin{tikzpicture}[scale=0.8] 

 \draw [|-|] (0,0)node[below=3pt]{$0$} -- (8,0)node[below=3pt]{$1$};
\draw (0,0) -- (4,0)node[below=12pt]{$T$};
\draw [->,dashed] (1,1) -- (7,1);
\draw (4.5,2)node[below]{better};
\end{tikzpicture}
\end{figure}

We can identify a candidate for a stable outcome with a measure on $T\times T$ with both marginals being $\nu$ since a positive measure of unmatched agents cannot occur in a stable outcome and no agent of type $t$ cares whether they are assigned contract $(t',1)$ or $(t',2)$.\\

In the only stable outcome of this example essentially every agent is matched to an agent of the same type. Formally, the only stable outcome is the unique outcome that is supported on the diagonal $\Delta=\{(t,t')\mid t=t'\}$. This outcome is obviously individually rational and stable. To see that no other outcome is stable, note that if the set $B= T\times T\setminus \Delta$ has positive measure, then from Definition \ref{def:outcome}.3 of an outcome, it follows that there is a set $B'\in T\times T$ with positive measure, such that for all $(t,t')$ in $ B'$ one has $t>t'$. Agents of matched type $(t,t')\in B'$ have an incentive to deviate by being matched with an agent of their own type instead, hence they form a pairwise block that arises at the outcome. 

\bigskip

\paragraph{A cyclic roommate problem:}

Let the types of agents be points on a circle. Formally, we let $T = [0, 1)$,
endowed with the metric $d$ given by
$d(x,y) = \min \{|x-y|, 1-|x-y|\}$.
Basically, we wrap the half-open unit interval $[0, 1)$ around a circle of circumference $1$, or,
equivalently, take the closed interval $[0, 1]$ and glue the endpoints together.\\

\begin{figure}[ht] 
\centering
\begin{tikzpicture}[scale=0.74] 
\draw[thick] (0,0) circle (2);

\draw[thick,|-|] (-6.2831,-2.5)node[below=0.2cm]{$0$} -- (6.28,-2.5)node[below=0.2cm]{$1$};
\draw[fill,color=black] (-3.2831+1.5707,-2.5) circle [radius=0.06]node[below=0.2cm]{$x$};
\draw[fill,color=black] (3.2831-1.5707,-2.5) circle [radius=0.06]node[below=0.2cm]{$y$};

\draw[<->] (-1.7,1.7) arc (135:45:2.4);
\draw (0,2.4)node[above]{$d(x,y)$};

\draw[fill,color=black] (-1.416,1.416) circle [radius=0.06]node[below right]{$x$};
\draw[fill,color=black] (1.416,1.416) circle [radius=0.06]node[below left]{$y$};
\end{tikzpicture}
\end{figure} 

Let the contracts be uniquely defined by the two types involved, $Y=T$, and every agent can only be matched to one other agent, hence $\chi(t)=X^{(\leq 1)}$. Let $\alpha\in [0,1]$ be fixed. Every agent of type $t$ aims to be matched with an agent who is $\alpha$ shifted clockwise away and measures the attractiveness of any other partner in terms of closeness to this ideal partner. Being unmatched is the worst that can happen. In this example, again, it doesn't matter which role in a contract an agent signs. Formally, let the utility for type $t$ to be matched with type $t'$ be given by $u_t\big((t',1)\big)=u_t\big((t',2)\big)=-d(t+\alpha,t')$ and $u_t(\emptyset)=-1$.\\

We can identify a candidate for a stable outcome with a measure on $T\times T$ with both marginals being $\nu$ since unmatched types cannot occur in a stable outcome and no agent of type $t$ cares whether they are assigned contract $(t,(t',1))$ or $(t,(t',2))$. Depending on $\alpha$, different outcomes can be stable. We distinguish four cases:\\
\begin{itemize}

\item[Case 1] Let $\alpha<\nicefrac{1}{4}$. In the only stable outcome, essentially every agent is matched to an agent of the same type. Formally, the only stable outcome is the unique outcome that is supported on the diagonal $\Delta=\{(t,t')\mid t=t'\}$. \\

This outcome is stable: The outcome is obviously individually rational. Assume a (pairwise) block arises, then there exists a positive measure of 2-samples consisting of two matched types $(t,t)$ and $(t',t')$ with $t\neq t'$ s.t. $u_t(t')>u_t(t)$, meaning $t'$ is closer to $t+\alpha$ than $t$ itself, i.e.  $d(t',t+\alpha)<d(t,t+\alpha)=\alpha$, and $u_{t'}(t)>u_{t'}(t')$, hence $d(t,t'+\alpha)<d(t,t+\alpha)=\alpha$. This never occurs since $\alpha<\nicefrac{1}{4}$.\\

To see that the outcome is the only stable outcome, note that if the set $B= T\times T\setminus \Delta$ has positive measure, then it follows from Definition \ref{def:outcome}.3 of an outcome that there is a set $B'\subseteq T\times T$ with positive measure, such that for all $(t,t')\in B'$, the type $t'$ is more than 180 degrees further clockwise shifted away from $t$, i.e. $t'-t \mod 1>\nicefrac{1}{2}$, in particular more shifted away than $2\alpha$ since $\alpha<\nicefrac{1}{4}$. Thus $t'$ is further away from $t+\alpha$ than $t$ itself. Hence, agents of matched type $(t,t')\in B'$ have an incentive to deviate by being matched with an agent of their own type instead and form a pairwise block that arises at the outcome. 
\item[Case 2] Let $\nicefrac{1}{4}<\alpha<\nicefrac{1}{2}$.
There is a stable outcome in which every agent is matched to an agent on the other side of the circle. Formally, this stable outcome is the unique outcome that is supported on the set $O=\{(t,t')\mid t'=t+\nicefrac{1}{2} \mod 1\}$.\\

To see that this outcome is stable, note first that it is individually rational. Now, assume a (pairwise) block arises, then similar to above, the set of 2-samples that form a block have positive measure. Such samples are matched types $(t,t+\nicefrac{1}{2} \mod 1)$ and $(t',t'+\nicefrac{1}{2} \mod 1)$, such that both can improve by deviating.
 Either $t$ is more than 180 degrees clockwise away from $t'$, or vice versa. Say, we are in the first case. Then $t'$ would be worse off with $t$ than with $t'+\nicefrac{1}{2}$.\\


\item[Case 3] Let $\alpha>\nicefrac{3}{4}$. It follows directly from the symmetry, that the stable outcomes are exactly the same as in case 1 .
\item[Case 4]  Let $\alpha=\nicefrac{1}{4}$ or $\alpha=\nicefrac{3}{4}$.
There are stable outcomes given by $\mu(t,t')>0$ if and only if $t'=t+\nicefrac{1}{2}$ or $t'=t$. This is an obvious consequence of the cases above, since for these $\alpha$ agents are indifferent between being matched to their own type or to the type on the opposite side of the circle. 

\end{itemize}




\section{Discussion}

We develop a distributional model of many-to-many matching markets with large type spaces and show that stable outcomes exist whenever stable outcomes exist in corresponding  finite models. In particular, we show that tree-stable and, hence, pairwise stable outcomes exist.\\

Similar arguments could be used to  transfer existence results for other stability concepts to large markets, such as path- or trail-stability, strong group stability, or the core. Our work should be understood as a toolkit that can be generally applied to mechanically transfer existence results for matching markets and some of the structure of stable matchings. Continuous models often allow for more convenient tools, such as necessary first-order conditions. Often, one can learn much from necessary conditions implied by a solution concept. Our existence results ensure that this is not an empty exercise about nothing.
\\ 


In future research, our model could be extended to also include externalities or match-dependent contracts. Indeed, the topological fixed-point methods used by \citet{jagadeesan2021stability} allow for indifferences, and the appropriate correspondence could be modified to include externalities, as has been done in a simpler setting by \citet{greinecker2018pairwise}.  





\appendix

\section{Polish Multispaces}
 \label{app:multi}
A \emph{multiset} is a function with finite graph whose range consists of strictly positive integers. We let $\text{supp}(m)$ be the image of the projection of $m$ onto the first coordinate. If $m$ is a multiset, we let $|m|=\sum_{(x,l)\in m}l$ be its \emph{cardinality}. If the support of a multiset is contained in a set $X$, we call the multiset also an \emph{$X$-multiset}. We can identify the set of $X$-multisets with cardinality $n$ with the quotient of $X^n$ under permutations of indices. We write $X^{(n)}$ for the set of $X$-multisets with cardinality $n$. By abuse of notation, we identify a permutation $\sigma$ of the set $\{1,2,\ldots,n\}$, an element of the symmetric group $S_n$, with the induced bijection $\sigma:X^n\to X^n$ obtained by switching coordinates accordingly. If $X$ is a topological space, $\sigma:X^n\to X^n$ is a homeomorphism. We define an equivalence relation on $X^n$ by letting $x\sim x'$ if and only if $x=\sigma(x')$ for some $\sigma\in S_n$. Abusing notation, we denote the resulting quotient set by $X^{(n)}$. There is an obvious way to identify the space of multisets $X^{(n)}$ with this quotient set. This identification allows us to topologize $X^{(n)}$. If $X$ is a topological space, we endow $X^{(n)}$ with the corresponding quotient topology.

\begin{lemma}The quotient mapping $q:X^n\to X^{(n)}$ is open.
\end{lemma}
\begin{proof}
Let $U$ be an open subset of $X^n$ and let $V=\bigcup_{\sigma\in S_n}\sigma(U)$. Then $q(U)=q(V)$ and $q^{-1}(q(U))=V$, so $q(U)$ is an open subset of $X^{(n)}$.
\end{proof}

\begin{proposition}\label{quotientop}
Let $X$ be a Hausdorff space. Then a net $\big\langle [x_\alpha]\big\rangle$ in $X^{(n)}$ converges to $[x]\in X^{(n)}$ if and only if there is a net $\langle \sigma_\alpha\rangle$ in $S_n$ such that $\langle \sigma_\alpha x_\alpha\rangle$ converges to $x\in X^n$.
\end{proposition}
\begin{proof}If there is a net $\langle \sigma_\alpha x_\alpha\rangle$ in $X^n$ converging to $x\in X^n$, then $\big\langle [\sigma_\alpha x_\alpha]\big\rangle=\big\langle [ x_\alpha]\big\rangle$ converges to $[x]$ by the continuity of the quotient mapping.


For the other direction, assume that  a net $\langle [x_\alpha]\rangle$ in $X^{(n)}$ converges to $[x]\in X^{(n)}$. For each $x'\in [x]$, choose an open neighborhood $O(x')\subseteq X^n$ of $x'$ such that the family $\{ O(x')\mid x'\in [x]\}$ is disjoint. That this is possible follows from $X$ and, therefore, $X^n$ being Hausdorff. Let 
\[O=\bigcap\big\{\sigma(O(x'))\mid x\in \sigma(O(x')), x'\in [x],\sigma\in S_n\big\}.\]
Then $O$ is an open neighborhood of $x$ and the open set $\sigma(O)$ contains exactly one point of $[x]$ for each $\sigma\in S_n$. Since the quotient mapping $q:X^n\to X^{(n)}$ is open and $x\in O$,  $q(O)$ is an open neighborhood of $q(x)=[x]$. Since $\langle [x_\alpha]\rangle$ in $X^{(n)}$ converges to $[x]\in X^{(n)}$, the net $\langle [x_\alpha]\rangle$ must be eventually in $q(O)$. So $q^{-1}(q(O))=\bigcup_{\sigma\in S_n}\sigma(O)$ must contain $q^{-1}([x_\alpha])$ for $\alpha$ large enough. And for $\alpha$ large enough, there will be exactly one point $x'_\alpha\in q^{-1}([x_\alpha])\cap O$. For such $\alpha$, fix $x'_\alpha$ accordingly and let $\sigma_\alpha$ be the unique element of $S_n$ such that $\sigma_\alpha x_\alpha=x_\alpha'$. Specify $\sigma_\alpha$ arbitrarily for the remaining $\alpha$'s. We claim that $\langle \sigma_\alpha x_\alpha\rangle$, which eventually coincides with $\langle x_\alpha'\rangle$ converges to $x$. Indeed, let $W$ be an open neighborhood of $x$ in $X^n$. For $\alpha$ large enough, $[x_\alpha]$ must lie in $q(W\cap O)$ since $q$ is an open mapping and $x_\alpha\sim\sigma_\alpha x_\alpha=x'_\alpha\in O$ by construction. But then, $q^{-1}([ x_\alpha])\subseteq\bigcup_{\sigma\in S_n}\sigma(W\cap O)$. Since the sets of the form $\sigma(O\cap W)\subseteq\sigma(O)$ are disjoint, we must have $x'_\alpha\in O\cap W\subseteq W$. So $\langle x'_\alpha\rangle$ converges to $x$ and so does, therefore $\langle \sigma_\alpha x_\alpha\rangle$.
\end{proof}

\begin{proposition}If $X$ is a Polish space, then $X^{(n)}$ is a Polish space too.
\end{proposition}  
\begin{proof}Let $d$ be a compatible complete metric topologizing $X$ and let $d_\infty$ be the metric on $X^n$ given by $d_\infty(x,x')=\max_{i=1}^n d(x_i,x'_i)$. Under this metric, elements of $S_n$ are isometries. We define a function $d^*:X^{(n)}\times X^{(n)}\to\mathbb{R}$ by 
\[d^*\big([x],[x']\big)=\min_{\sigma\in S_n} d_\infty(x,\sigma x').\]
Note that the choice of the representatives $x$ and $x'$ is entirely irrelevant, so this function is well-defined. It is easily shown using the group structure of $S_n$ that $d^*$ is a metric. That $d^*\big([x],[x']\big)=0$ is equivalent to $[x]=[x']$ follows from $x\sim x'$ being equivalent to $x=\sigma x'$ for some $\sigma\in S_n$. To note that $d^*$ is symmetric, observe that
 \[d_\infty(x,\sigma x')=d_\infty(\sigma^{-1}x,\sigma^{-1}\sigma x')=d_\infty(\sigma^{-1} x,x')=d_\infty(x',\sigma^{-1}x).\]
The triangle inequality follows from
 \[d_\infty(x,\sigma x')+d_\infty(x',\sigma' x'')=d_\infty(x,\sigma x')+d_\infty(\sigma x' ,\sigma\sigma' x'')\geq d_\infty(x,\sigma\sigma' x'').\]

Next, we show $d^*$ is complete. Let $\langle [x_n]\rangle$ be a Cauchy sequence. Recursively, define a sequence $\langle x_n'\rangle$ such that $x_1'=x_1$, and $x_n'\sim x_n$ and $d_\infty(x_n',x_{n+1}')=d^*\big([x_n],[x_{n+1}]\big)$ for all $n$. The sequence $\langle x_n'\rangle$ is a Cauchy sequence in the complete metric space $(X^n,d_\infty)$ and converges to some $x\in X^n$. Clearly, $d^*\big([x_n],[x]\big)\leq d_\infty(x_n,x)$, so $\langle [x_n]\rangle$ converges to $[x]$.\smallskip

We have to show that $d^*$ actually metrizes the quotient topology on $X^{(n)}$.  We are going to use Proposition \ref{quotientop}. Let $\langle [x_\alpha]\rangle$ be a net in $X^{(n)}$ converging to $[x]$ in the quotient topology. There must be a net $\langle \sigma_\alpha\rangle$ in $S_n$ such that $\langle \sigma_\alpha x_\alpha\rangle$ converges to $x\in X$. Consequently, $d_\infty(x_\alpha,x)$ converges to $0$ and, therefore, so does $d_\infty\big([x_\alpha],[x]\big).$ Consequently, the topology induced by $d^*$ on $X^{(n)}$ is at least as coarse as the quotient topology. 

For the other direction, let $O$ be a nonempty open set in the quotient topology and $[x]\in O$. The set $q^{-1}(O)$ is open in $X^n$ and $[x]\subseteq q^{-1}(O)$. For each $x'\in [x]$, let $B_\epsilon(x')$ be the corresponding $d_\infty$-ball with radius $\epsilon$ and center $x'$. For $\epsilon>0$ small enough, 
$\bigcup_{x'\in [x]}B_\epsilon(x')\subseteq q^{-1}(O)$. For such $\epsilon>0$ let $B^*_\epsilon([x])$ be the corresponding $d_\infty$-ball with this radius around $[x]$. It is clear that $B^*_\epsilon([x])\subseteq O$, so the topology induced by $d^*$ is at least as fine as the quotient topology.\smallskip
 
Finally, to show that  $X^{(n)}$ is separable, let $D\subseteq X^n$ be a countable dense set. Since $q$ is continuous, the countable set $q(D)$ is dense in $q(X^n)=X^{(n)}$.
\end{proof}

We define a function $*:\mathcal{M}\big(X^{(n)}\big)\to\mathcal{M}\big(X\big)$ as follows. There exists a universally measurable function $s:X^{(n)}\to X^n$ such that $q\circ s$ is the identity on $X^{(n)}$ (Aliprantis and Border, Corollary 18.23). For $i=1,\ldots,n$, let $\pi_i:X^n\to X$ be the corresponding projection.
Now define the $*$-mapping by 
\[\tau^*(E)=\sum_{i=1}^n \tau\circ s^{-1}\big(\pi_i^{-1}(E)\big)\]
for every Borel set $E\subseteq X$. 

\begin{lemma}\label{counting}
The function $*:\mathcal{M}\big(X^{(n)}\big)\to\mathcal{M}\big(X\big)$ is continuous.
\end{lemma}
\begin{proof}We need to show that the function $\tau\mapsto \int f~\mathrm d\tau^*$ is continuous for every bounded continuous function $f:X\to\mathbb{R}$.
Now, \[\int f~\mathrm d\tau^*=\int f~\mathrm d \sum_{i=1}^n \tau\circ s^{-1}\circ\pi_i^{-1}=\sum_{i=1}^n\int f~\mathrm d\tau\circ s^{-1}\circ\pi_i^{-1}.\]
By the change of variables formula for pushforward measures, this is equal to
\[\sum_{i=1}^n\int f\circ\pi_i\circ s~\mathrm d\tau=\int \sum_{i=1}^n f\circ\pi_i\circ s~\mathrm d\tau.\]
By the definition of the weak topology on $\mathcal{M}(X^{(n)})$, it suffices to show that the bounded function \[\sum_{i=1}^n f\circ\pi_i\circ s:X^{(n)}\to\mathbb{R}\] is continuous. By the universal property of the quotient topology, it suffices to show that its composition with $q$, 
\[\sum_{i=1}^n f\circ\pi_i\circ s\circ q:X^n\to\mathbb{R},\] is continuous.
Now, for $(x_1,\ldots,x_n)\in X^n$, we have \[(s\circ q) (x_1,\ldots,x_n)=(x_{\sigma^{-1} (1)},\ldots,x_{\sigma^{-1}(n)})\] for some $\sigma\in S_n$. Consequently, \[\bigg(\sum_{i=1}^n f\circ\pi_i\circ s\circ q\bigg)(x_1,\ldots,x_n)=\sum_{i=1}^n f(x_{\sigma^{-1}(i)})=\sum_{i=1}^n f(x_i).\]
Since the function $(x_1,\ldots,x_n)\mapsto \sum_{i=1}^n f(x_i)$ is clearly continuous, we are done.
\end{proof}

\section{Proofs}
\label{sec:proofs}




\begin{proof}[Proof of Lemma \ref{Lemmacompact}]
Clearly, the set of all pairs of such measures $(\nu,\mu)$ on $T\times \multisetDef$ with the $T$-marginal of $\mu$ being $\nu$ and supported on the closed graph of $\chi$ is closed by the continuity of the marginal function and the Portmanteau theorem, respectively. It remains to show that the third condition in the definition of an outcome defines a closed set. Since the marginal mapping is continuous and the function $\mu\mapsto\mu^*$ is continuous by Lemma \ref{counting}, it suffices to show that the function that maps a measure $\tau$ on $\multisetDef$ to the measure $\tau\big(\cdot\times\{i\}\big)$ is continuous for $i=1,2$. This follows from the Portmanteau theorem. For if $E$ is a continuity set with respect to the $Y$-marginal of $\tau$, then $E\times\{i\}$ is a continuity set too since the boundary of $\{i\}$ is empty. Indeed,\[\partial\big(E\times \{i\}\big)=\big(\partial E\times\textnormal{cl}(\{i\})\big)\cup\big(\textnormal{cl}(E)\times \partial\{i\}\big)=\partial E\times{\{i\}}\subseteq\partial E\times\{1,2\},\] and the last set has $\tau$-measure zero by assumption.
\end{proof}

\begin{proof}[Proof of Lemma \ref{tight}]
It suffices to show that the set $\{(\nu_n,\mu_n)\}$ is relatively compact, which is equivalent to being tight by Prohorov's theorem. The result follows then from Lemma \ref{Lemmacompact}.

Let $\epsilon>0$. Since $\langle \nu_n\rangle$ converges, the set of its terms is relatively compact and, hence, tight. In particular, there exists a compact set $K_\epsilon\subseteq T$ such that $\nu_n(T\setminus K_\epsilon)<\epsilon$ for all $n$. As we have argued in the main text, the existence of an upper hemicontinuous and compact-valued correspondence $P:T\rightrightarrows T$ such that $\bar{\chi}(t,t')=\emptyset$ for $t'\notin P(t)$ implies that the correspondence $t\mapsto X_t$ is upper hemicontinuous and compact-valued. This in turn implies that the correspondence $\chi$ is upper hemicontinuous and compact-valued.\footnote{Use the continuity of the quotient mapping from $n$-tuples to $n$-multisets, with \citet[Theorem 17.23, Theorem 17.27, and Theorem 17.28]{AliprantisBorder2003}.} Therefore, the set $\chi(K_\epsilon)=\bigcup_{t\in K_\epsilon}\chi(t)$ is compact by \citet[Theorem 17.8]{AliprantisBorder2003}. Now, using the three conditions defining an outcome, we have 
\[1-\epsilon<\nu_n(K_\epsilon)=\mu_n\big(K_\epsilon\times \multisetDef \big)=\mu_n\big(K_\epsilon\times \chi(K_\epsilon)).\]
So the set $K_\epsilon\times \chi(K_\epsilon)$ is compact, and $\mu_n(K_\epsilon\times \chi(K_\epsilon))>1-\epsilon$ for all $n$. So $\{(\nu_n,\mu_n)\}$ is tight and, by Prohorov's theorem, relatively compact.
\end{proof}

\begin{proof}[Proof of Lemma \ref{open}]
If an n-sample $t^n$ is a $G$-block, then any n-sample $\tilde{t}^n$ close enough to $t^n$ is a $G$-block: By continuity of the utility function $u:T\times \multisetDef\to\mathbb{R}$ it follows that the set of possible contract choices $x_{(k,l)}$ in Definition \ref{def:block} such that the corresponding $\tilde{\mathcal{Z}}^j$ satisfies the `desirability property' (there exists a $\tilde{\mathcal{W}}^j\msubseteq \tilde{\mathcal{Z}}^j\mcup \tilde{\mathcal{Y}}^j$ such that $\tilde{\mathcal{Z}}^j\msubseteq \tilde{\mathcal{W}}^j$, and $u_{\tilde{t}_j}(\tilde{\mathcal{W}}^j)>u_{\tilde{t}_j}(\tilde{\mathcal{Y}}^j)$ for each $1\leq j \leq n$)
is open. From the lower hemicontinuity of $\bar{\chi}$ and Proposition \ref{quotientop} it thus follows, that for all $(k,l)\in G$ there exists a choice of contracts $\tilde{x}_{(k,l)}$, such that there exists such a $\tilde{\mathcal{W}}^j$ that lies in  $\chi(\tilde{t}_j)$ for all $1\leq j\leq n$; i.e. there is a set of new contract choices $\tilde{\mathcal{W}}^j$ including the new contracts $\tilde{x}_{(k,l)}$ for all involved matched types $\tilde{t}^j$ that is not only desirable but also feasible. 
\end{proof}

\begin{proof}[Proof of Lemma \ref{lem:TYfinite}]
In comparison to the model in \citet{jagadeesan2021stability}, contracts in our model have specified sides, do not uniquely specify the pair of types involved, agents of a type can sign contracts with agents of the same type, and we allow for multisets. While in large markets the assumptions of our model are either necessary or more natural, in finite markets one model can easily be transformed into the other by modifying the contract and type sets since we do not have to preserve any nontrivial topological structure. 
To apply the existence theorem we can construct for each matching problem in our model a matching problem in the model of \citet{jagadeesan2021stability}. The contract space can be expanded to $X\times T\times T$ to have uniquely specified types for each contract and by enriching the space of contracts the sides of contracts can be encoded too. We can expand the set of types to $T^2$ to make sure that agents need not sign contracts with the same type. To transfer a model with multisets with a maximum of $N$ contracts to a model with sets with a maximum of $N$ contracts we can construct a new contract space $X\times N\times N$, such that agents can sign different versions of a contract instead of signing the contract several times. 
\end{proof}








\bibliographystyle{chicago}

\end{document}